\documentclass[a4paper,reqno]{amsart}
\usepackage{amsmath,amssymb,mathrsfs}
\usepackage{txfonts,bm,pifont}
\usepackage{cite}
\usepackage[dvipdfm,bookmarks=true,bookmarksnumbered=true,colorlinks=true]{hyperref}
\usepackage{comment}
\usepackage[dvipdfmx]{graphicx}
\PassOptionsToPackage{svgnames}{xcolor}
\usepackage{tikz}
\usetikzlibrary{calc}
\usetikzlibrary{decorations.markings,decorations.pathmorphing}
\newtheorem{theorem}{Theorem}[section]

\newtheorem{lemma}[theorem]{Lemma}
\newtheorem{remark}[theorem]{Remark}

\numberwithin{equation}{section}
\newcommand{\wutilde}[1]{\vrule depth 0pt width 0pt%
{\raise0.8pt\hbox{$\smash{{\mathop{#1} \limits_{\displaystyle\widetilde{}}}}$}}}
\newcommand{\ol}[1]{\overline{#1}}
\newcommand{\ul}[1]{\underline{#1}}
\newcommand{\al}{\alpha}
\newcommand{\be}{\beta}

\newcommand{\ep}{\epsilon}

\newcommand{\PDE}{P$\Delta$E}
\newcommand{\bbZ}{\mathbb{Z}}

\newcommand{\bbC}{\mathbb{C}}

\newcommand{\bi}{{\overline{i}}}
\newcommand{\bj}{{\overline{j}}}
\makeatletter
\long\def\@makecaption#1#2{
 \vskip 10pt
 \setbox\@tempboxa\hbox{#1. #2}
 \ifdim \wd\@tempboxa >\hsize #1. #2\par \else \hbox
to\hsize{\hfil\box\@tempboxa\hfil}
 \fi}
\makeatother
\begin{document}
%%%%%%%%%%%%%%%%%%%%%%%%%%%%%%%%%%%%%%%%%%%%%%%%%%%%%%%
%%% title
%%%%%%%%%%%%%%%%%%%%%%%%%%%%%%%%%%%%%%%%%%%%%%%%%%%%%%%
\title[]
{Reduction of lattice equations to the Painlev\'e equations: P$_{\rm IV}$ and P$_{\rm V}$}
\author{Nobutaka Nakazono}
\address{Department of Physics and Mathematics, Aoyama Gakuin University, Sagamihara, Kanagawa 252-5258, Japan.}
\email{nobua.n1222@gmail.com}
%%%%%%%%%%%%%%%%%%%%%%%%%%%%%%%%%%%%%%%%%%%%%%%%%%%%%%%
%% Abstract
%%%%%%%%%%%%%%%%%%%%%%%%%%%%%%%%%%%%%%%%%%%%%%%%%%%%%%%
\begin{abstract}
In this paper, we construct a new relation between ABS equations and Painlev\'e equations.
Moreover, using this connection we construct the difference-differential Lax representations of the fourth and fifth Painlev\'e equations.
\end{abstract}
%%%%%%%%%%%%%%%%%%%%%%%%%%%%%%%%%%%%%%%%%%%%%%%%%%%%%%%
%%%%%%%%%%%%%%%%%%%%%%%%%%%%%%%%%%%%%%%%%%%%%%%%%%%%%%%

\maketitle

%%%%%%%%%%%%%%%%%%%%%%%%%%%%%%%%%%%%%%%%%%%%%%%%%%%%%%%%%%%
%% 1. Introduction
%%%%%%%%%%%%%%%%%%%%%%%%%%%%%%%%%%%%%%%%%%%%%%%%%%%%%%%%%%%
\section{Introduction}
In recent works by Joshi-Nakazono-Shi \cite{JNS2015:MR3403054,JNS2014:MR3291391,JNS:paper4,JNS2016:MR3584386}, 
the mathematical connection between two longstanding classifications of integrable systems in different dimensions, 
one by Adler-Bobenko-Suris (ABS equations) \cite{ABS2003:MR1962121,HietarintaJ2005:MR2217106,ABS2009:MR2503862,BollR2011:MR2846098,BollR2012:MR3010833,BollR:thesis} and 
the other by Okamoto and Sakai (Painlev\'e and discrete Painlev\'e equations) \cite{SakaiH2001:MR1882403,OkamotoK1979:MR614694}, 
have been investigated by using their lattice structures.
Moreover, a comprehensive method of constructing Lax representations of discrete Painlev\'e equations using this connection was provided in \cite{JNS:paper3} and demonstrated in \cite{JNS:paper3,JNS2016:MR3584386}.
The whole picture of the connection between the ABS equations and the discrete Painlev\'e equations has been gradually revealed,
but that between the ABS equations and the (differential) Painlev\'e equations was missing, that is,
there is still a great distance between them.
In the present paper, we fill this gap by erecting a bridge from the ABS equations to the Painlev\'e equations.

A hierarchy of nonlinear ordinary differential equations (ODEs) found by Noumi-Yamada in \cite{NY1998:MR1676885,book_NoumiM2004:MR2044201} is sometimes referred to as {\it NY-system}.
It is well known that NY-system contains the fourth and fifth Painlev\'e equations (P$_{\rm IV}$ and P$_{\rm V}$) and has an $A$-type affine Weyl group symmetry.

In this paper, we show that a system of ABS equations can be reduced to NY-system by a periodic type reduction.
Through this connection we construct the difference-differential Lax representations of P$_{\rm IV}$ and P$_{\rm V}$ (see Theorems \ref{theorem:p4_Lax} and \ref{theorem:p5_Lax}).
Moreover, we obtain remarkable results that the dependent variable of the system of ABS equations can be reduced to the Hamiltonians of P$_{\rm IV}$ and P$_{\rm V}$ (see Theorems \ref{theorem:p4_hamiltonian} and \ref{theorem:p5_hamiltonian}).

%%%%%%%%%%%%%%%%%%%%%%%%%%%%%%%%%%%%%%%%%%%%%%%%%%%%%%%%%%%
%% 1.1 The fourth and fifth Painlev\'e equations
%%%%%%%%%%%%%%%%%%%%%%%%%%%%%%%%%%%%%%%%%%%%%%%%%%%%%%%%%%%
\subsection{The fourth and fifth Painlev\'e equations}
In this paper, we focus on the following Painlev\'e equations: 
\begin{equation}\label{eqn:intro_P4}
\text{P$_{\rm IV}$: }
\begin{cases}
 f_0'=f_0(f_2-f_1)+3a_0,\\
 f_1'=f_1(f_0-f_2)+3a_1,\\
 f_2'=f_2(f_1-f_0)+3a_2,
\end{cases}
\end{equation}
where
\begin{equation}\label{eqn:intro_fa_P4}
 f_0+f_1+f_2=3t,\quad
 a_0+a_1+a_2=1,
\end{equation}
and 
\begin{equation}\label{eqn:intro_P5}
\text{P$_{\rm V}$: }
\begin{cases}
 2t f_0'=f_0f_2(f_3-f_1)+4a_0f_2+2(1-2a_2)f_0,\\
 2t f_1'=f_1f_3(f_0-f_2)+4a_1f_3+2(1-2a_3)f_1,\\
 2t f_2'=f_2f_0(f_1-f_3)+4a_2f_0+2(1-2a_0)f_2,\\
 2t f_3'=f_3f_1(f_2-f_0)+4a_3f_1+2(1-2a_1)f_3,
\end{cases}
\end{equation}
where
\begin{equation}\label{eqn:intro_fa_P5}
 f_0+f_2=f_1+f_3=2t,\quad
 a_0+a_1+a_2+a_3=1.
\end{equation}
Note that in both cases, $f_i=f_i(t)$ are dependent variables, $t$ is the independent variable, $a_i$ are complex parameters and $'$ denotes ${\rm d}/{\rm d}t$.
The polynomial Hamiltonians of P$_{\rm IV}$ and P$_{\rm V}$ \cite{OkamotoK1980:MR581468,OkamotoK1980:MR596006} are respectively given by $3\sqrt{-3}\,h_{\rm IV}$ and $16\,h_{\rm V}$, where
\begin{align}
 &h_{\rm IV}=\frac{1}{3\sqrt{-3}}\,\Big(f_0f_1f_2-(a_1-a_2)f_0-(a_1+2a_2)f_1+(2a_1+a_2)f_2\Big),\label{eqn:intro_hamiltonian_P4}\\
 \begin{split}\label{eqn:intro_hamiltonian_P5}
 &h_{\rm V}=\frac{1}{16}\,\Big(f_0 f_1 f_2 f_3-(a_1+2 a_2-a_3) f_0 f_1-(a_1+2 a_2+3 a_3) f_1 f_2\\
 &\hspace{4em}+(3 a_1+2 a_2+a_3) f_2 f_3-(a_1-2 a_2-a_3) f_3 f_0+4 (a_1+a_3)^2\Big).
 \end{split}
\end{align}

%%%%%%%%%%%%%%%%%%%%%%%%%%%%%%
%% Remark
%%%%%%%%%%%%%%%%%%%%%%%%%%%%%%
\begin{remark}
{\rm P$_{\rm IV}$} \eqref{eqn:intro_P4} can be rewritten as the following ``standard" symmetric form given in 
{\rm \cite{NY1999:MR1684551,NY1998:MR1666847,book_NoumiM2004:MR2044201}:}
\begin{equation}\label{eqn:intro_P4_2}
\begin{cases}
 \dfrac{{\rm d}F_0(s)}{{\rm d}s}=F_0(s)(F_1(s)-F_2(s))+a_0,\\[0.5em]
 \dfrac{{\rm d}F_1(s)}{{\rm d}s}=F_1(s)(F_2(s)-F_0(s))+a_1,\\[0.5em] 
 \dfrac{{\rm d}F_2(s)}{{\rm d}s}=F_2(s)(F_0(s)-F_1(s))+a_2,
\end{cases}
\end{equation}
by the following replacements:
\begin{equation}\label{eqn:replacement_fundamental_P4}
 -\sqrt{-3} F_i(s)=f_i(t),\quad i=0,1,2,\quad s=\sqrt{-3} t.
\end{equation}
Also, we can express {\rm P$_{\rm V}$} \eqref{eqn:intro_P5} in the following ``standard" symmetric form given in 
{\rm \cite{book_NoumiM2004:MR2044201,NY1998:MR1666847}:}
\begin{equation}\label{eqn:intro_P5_2}
\begin{cases}
 \dfrac{{\rm d}F_0(s)}{{\rm d}s}=F_0(s)F_2(s)(F_1(s)-F_3(s))+a_0F_2(s)+\dfrac{1-2a_2}{2}F_0(s),\\[0.5em]
 \dfrac{{\rm d}F_1(s)}{{\rm d}s}=F_1(s)F_3(s)(F_2(s)-F_0(s))+a_1F_3(s)+\dfrac{1-2a_3}{2}F_1(s),\\[0.5em]
 \dfrac{{\rm d}F_2(s)}{{\rm d}s}=F_2(s)F_0(s)(F_3(s)-F_1(s))+a_2F_0(s)+\dfrac{1-2a_0}{2}F_2(s),\\[0.5em]
 \dfrac{{\rm d}F_3(s)}{{\rm d}s}=F_3(s)F_1(s)(F_0(s)-F_2(s))+a_3F_1(s)+\dfrac{1-2a_1}{2}F_3(s),
\end{cases}
\end{equation}
by the following replacements:
\begin{equation}\label{eqn:replacement_fundamental_P5}
 F_i(s)=\frac{\sqrt{-1}}{2}\, f_i(t),\quad i=0,1,2,3,\quad s=2\log{t}.
\end{equation}
Moreover, by using the replacements \eqref{eqn:replacement_fundamental_P4} and \eqref{eqn:replacement_fundamental_P5}, Equations \eqref{eqn:intro_hamiltonian_P4} and \eqref{eqn:intro_hamiltonian_P5} can be rewritten as
\begin{align}
 &h_{\rm IV}=F_0F_1F_2+\frac{a_1-a_2}{3}F_0+\frac{a_1+2a_2}{3}F_1-\frac{2a_1+a_2}{3}F_2,\\
 \begin{split}
 &h_{\rm V}=F_0 F_1 F_2 F_3+\frac{a_1+2 a_2-a_3}{4} F_0 F_1+\frac{a_1+2 a_2+3 a_3}{4}F_1 F_2\\
 &\qquad-\frac{3 a_1+2 a_2+a_3}{4} F_2 F_3+\frac{a_1-2 a_2-a_3}{4} F_3 F_0+ \frac{(a_1+a_3)^2}{4},
 \end{split}
\end{align}
which are the Hamiltonians of Equations \eqref{eqn:intro_P4_2} and \eqref{eqn:intro_P5_2}, respectively.
\end{remark}
%%%%%%%%%%%%%%%%%%%%%%%%%%%%%%
%%%%%%%%%%%%%%%%%%%%%%%%%%%%%%

%%%%%%%%%%%%%%%%%%%%%%%%%%%%%%%%%%%%%%%%%%%%%%%%%%%%%%%%%%%
%% 1.2 Main results
%%%%%%%%%%%%%%%%%%%%%%%%%%%%%%%%%%%%%%%%%%%%%%%%%%%%%%%%%%%
\subsection{Main results}
\label{subsection:main_result}
In this section, we outline four main results of this paper.

Firstly, in \S \ref{subsection:P4} we prove the following theorems for P$_{\rm IV}$ \eqref{eqn:intro_P4}.
%%%%%%%%%%%%%%%%%%%%%%%%%%%%%%
%% Theorem
%%%%%%%%%%%%%%%%%%%%%%%%%%%%%%
\begin{theorem}\label{theorem:p4_hamiltonian}
The dependent variable of the system of ABS equations \eqref{eqn:standard_H1} with $n=2$
can be reduced to the Hamiltonian \eqref{eqn:intro_hamiltonian_P4}.
\end{theorem}
%%%%%%%%%%%%%%%%%%%%%%%%%%%%%%
%% Theorem
%%%%%%%%%%%%%%%%%%%%%%%%%%%%%%
\begin{theorem}\label{theorem:p4_Lax}
The Lax representation of {\rm P}$_{\rm IV}$ \eqref{eqn:intro_P4} is given by the following:
\begin{equation}
 \Phi(x+1,t)=A_{\rm IV}(x,t)\Phi(x,t),\quad
 \frac{\partial}{\partial t}\Phi(x,t)=B_{\rm IV}(x,t)\Phi(x,t),
\end{equation}
that is, the compatibility condition
\begin{equation}
 \frac{\partial}{\partial t}A_{\rm IV}(x,t)+A_{\rm IV}(x,t)B_{\rm IV}(x,t)=B_{\rm IV}(x+1,t)A_{\rm IV}(x,t),
\end{equation}
is equivalent to {\rm P}$_{\rm IV}$ \eqref{eqn:intro_P4}.
Here, 
\begin{align}
 \begin{split}
 A_{\rm IV}(x,t)
 =&\begin{pmatrix}
 1&f_1+\omega_0+\frac{(2a_1+a_2)t}{2}+\frac{3t}{2}x\\0&1
 \end{pmatrix}
 \begin{pmatrix}
 0&-3(a_1+a_2)-3x+\mu\\1&-f_2
 \end{pmatrix}\\
 &\quad\begin{pmatrix}
 0&-3a_1-3x+\mu\\1&-f_1
 \end{pmatrix}
 \begin{pmatrix}
 0&-3x+\mu\\1&f_2+\omega_0-\frac{(2a_1+a_2+3)t}{2}-\frac{3t}{2}x
 \end{pmatrix},
 \end{split}\\
 B_{\rm IV}(x,t)
 =&\begin{pmatrix}
 1&\omega_0+\frac{(2a_1+a_2)t}{2}+\frac{3t}{2}x\\0&1
 \end{pmatrix}
 \begin{pmatrix}
 0&-\omega_0'-\frac{2a_1-2a_0-t^2}{4}-\frac{3}{2}x+\mu\\1&-\omega_0-\frac{(2a_1+a_2)t}{2}-\frac{3t}{2}x
 \end{pmatrix},
\end{align}
where the variables $\omega_0$ and $\omega_0'$ are given by
\begin{align}
 &\omega_0=2\Big(f_0f_1f_2-(a_1-a_2)f_0-(a_1+2a_2)f_1+(2a_1+a_2)f_2\Big)-\frac{(3+2t^2)t}{6},\\
 &\omega_0'=-\frac{(f_0-t)(2t-f_0)-(f_1-t)(2t-f_1)-(f_2-t)(2t-f_2)+2(a_1-a_2)+1}{2},
\end{align}
and $\mu$ is an arbitrary complex constant.
\end{theorem}
%%%%%%%%%%%%%%%%%%%%%%%%%%%%%%
%%%%%%%%%%%%%%%%%%%%%%%%%%%%%%

Other two main results are for P$_{\rm V}$ \eqref{eqn:intro_P5} given by the following theorems.
%%%%%%%%%%%%%%%%%%%%%%%%%%%%%%
%% Theorem
%%%%%%%%%%%%%%%%%%%%%%%%%%%%%%
\begin{theorem}\label{theorem:p5_hamiltonian}
The dependent variable of the system of ABS equations \eqref{eqn:standard_H1} with $n=3$
can be reduced to the Hamiltonian \eqref{eqn:intro_hamiltonian_P5}.
\end{theorem}
%%%%%%%%%%%%%%%%%%%%%%%%%%%%%%
%% Theorem
%%%%%%%%%%%%%%%%%%%%%%%%%%%%%%
\begin{theorem}\label{theorem:p5_Lax}
The Lax representation of {\rm P}$_{\rm V}$ \eqref{eqn:intro_P5} is given by
\begin{equation}
 \Phi(x+1,t)=A_{\rm V}(x,t)\Phi(x,t),\quad
 \frac{\partial}{\partial t}\Phi(x,t)=B_{\rm V}(x,t)\Phi(x,t),
\end{equation}
that is, the compatibility condition 
\begin{equation}
 \frac{\partial}{\partial t}A_{\rm V}(x,t)+A_{\rm V}(x,t)B_{\rm V}(x,t)=B_{\rm V}(x+1,t)A_{\rm V}(x,t),
\end{equation}
is equivalent to {\rm P}$_{\rm V}$ \eqref{eqn:intro_P5}.
Here, 
\begin{align}
\begin{split}
 A_{\rm V}(x,t)
 =&
\begin{pmatrix}
 1 & \omega_1-f_0+\frac{(3 a_1+2 a_2+a_3+5) t}{2}+2 t x \\
 0 & 1 
\end{pmatrix}
\begin{pmatrix}
 0 & -4 (a_1+a_2+a_3)-4 x+\mu \\
 1 & -f_3 
\end{pmatrix}\\
&\quad\begin{pmatrix}
 0 & -4 (a_1+a_2)-4 x+\mu\\
 1 & -f_2 
\end{pmatrix}
\begin{pmatrix}
 0 & -4 a_1-4 x+\mu\\
 1 & -f_1 
\end{pmatrix}\\
&\quad\begin{pmatrix}
 0 & -4 x+\mu\\
 1 & -\omega_1-\frac{(3 a_1+2 a_2+a_3+1) t}{2}-2 t x 
\end{pmatrix},
\end{split}\\
 B_{\rm V}(x,t)
 =&\begin{pmatrix}
 1 &\omega_0+\frac{(3 a_1+2 a_2+a_3)t}{2}+2 t x\\
 0 & 1 
\end{pmatrix}
\begin{pmatrix}
 0 & -\omega_0'-\frac{2 (3 a_1+2 a_2+a_3)-t^2-2}{4}-2 x+\mu\\
 1 & -\omega_0-\frac{(3 a_1+2 a_2+a_3) t}{2}-2 t x
\end{pmatrix},
\end{align}
where
\begin{align}
 &\omega_1=\omega_0-\frac{f_2f_3-2(a_0+a_2)-t^2-1}{2t},\\
\begin{split}
 &\omega_0=\frac{f_0 f_1 f_2 f_3-(a_1+2 a_2-a_3) f_0 f_1-(a_1+2 a_2+3 a_3) f_1 f_2}{8t}\\
 &\qquad+\frac{(3 a_1+2 a_2+a_3) f_2 f_3-(a_1-2 a_2-a_3) f_3 f_0+4 (a_1+a_3)^2-t^4-6 t^2-1}{8t},
\end{split}\\
\begin{split}
 &\omega_0'=-\frac{4 (a_1+a_3)^2-1+2 (3+2 a_1+4 a_2-2 a_3) t^2+3 t^4-8 t (a_3 f_2+a_2 f_3)}{8 t^2}\\
 &\qquad-\frac{4 (a_1+a_3-t^2) f_2 f_3-2 t (f_2-f_3) f_2 f_3+{f_2}^2 {f_3}^2}{8 t^2},
\end{split}
\end{align}
and $\mu$ is an arbitrary complex constant.
\end{theorem}
%%%%%%%%%%%%%%%%%%%%%%%%%%%%%%
%%%%%%%%%%%%%%%%%%%%%%%%%%%%%%

The proofs of Theorems \ref{theorem:p5_hamiltonian} and \ref{theorem:p5_Lax} are given in \S \ref{subsection:P5}.

%%%%%%%%%%%%%%%%%%%%%%%%%%%%%%%%%%%%%%%%%%%%%%%%%%%%%%%%%%%
%% 1.3 Background
%%%%%%%%%%%%%%%%%%%%%%%%%%%%%%%%%%%%%%%%%%%%%%%%%%%%%%%%%%%
\subsection{Background}\label{subsection:background}
%%
%% Painleve equation
%%
The six Painlev\'e equations: P$_{\rm VI}$, \dots, P$_{\rm I}$ are nonlinear ODEs of second order
which have the Painlev\'e property, i.e., their solutions do not have movable branch points.
It is well known that the Painlev\'e equations, except for P$_{\rm I}$, have B\"acklund transformations, which collectively form affine Weyl groups. 
The following is the diagram of degenerations:
$$\begin{array}{ccccccc}
 {\rm P}_{\rm VI}~(D_4^{(1)})&\to&{\rm P}_{\rm V}~(A_3^{(1)})&\to&{\rm P}_{\rm III}~(2A_1^{(1)})&\\
 &&\downarrow&&\downarrow&&\\
 &&{\rm P}_{\rm IV}~(A_2^{(1)})&\to&{\rm P}_{\rm II}~(A_1^{(1)})&\to&{\rm P}_{\rm I}
\end{array}$$
where the symbols inside the parentheses indicate the types of affine Weyl groups.
(See \cite{OKSO2006:MR2277519,OkamotoK1987:MR927186,OkamotoK1987:MR914314,OkamotoK1987:MR916698,OkamotoK1986:MR854008,book_NoumiM2004:MR2044201,SakaiH2001:MR1882403,KNY2015:arXiv150908186K} for the details.)

%%
%% ABS equation
%%
In \cite{ABS2003:MR1962121,HietarintaJ2005:MR2217106,ABS2009:MR2503862,BollR2011:MR2846098,BollR2012:MR3010833,BollR:thesis}, 
Adler-Bobenko-Suris (ABS) {\it et al.} classified polynomials $P$, say, of four variables into 11 types:
Q4, Q3, Q2, Q1, H3, H2, H1, D4, D3, D2, D1.
The first four types, the next three types and the last four types are collectively called $Q$-, $H^4$- and $H^6$-types, respectively.
The resulting polynomial $P$ satisfies the following properties.
\begin{description}
\item[(1) Linearity]
$P$ is linear in each argument, i.e., it has the following form:
\begin{equation}
 P(x_1,x_2,x_3,x_4)=A_1x_1x_2x_3x_4+\cdots+A_{16},
\end{equation}
where coefficients $A_i$ are complex parameters.
\item[(2) 3D consistency and tetrahedron property]
There exist a further seven polynomials of four variables: $P^{(i)}$, $i=1,\dots,7$, which satisfy property {\bf(1)} 
and a cube $C$ on whose six faces the following equations are assigned:
\begin{subequations}
\begin{align}
 &P(x_0,x_1,x_2,x_{12})=0,
 &&P^{(1)}(x_0,x_2,x_3,x_{23})=0,\\
 &P^{(2)}(x_0,x_3,x_1,x_{31})=0,
 &&P^{(3)}(x_3,x_{31},x_{23},x_{123})=0,\\
 &P^{(4)}(x_1,x_{12},x_{31},x_{123})=0,
 &&P^{(5)}(x_2,x_{23},x_{12},x_{123})=0,
\end{align}
\end{subequations}
where the eight variables $x_0$, \dots, $x_{123}$ lie on the vertices of the cube (see Figure \ref{fig:cube}),
in such a way that 
$x_{123}$ can be uniquely expressed in terms of the four variables $x_0$, $x_1$, $x_2$, $x_3$
({\it 3D consistency})
and moreover the following relations hold ({\it tetrahedron property}):
\begin{equation}
 P^{(6)}(x_0,x_{12},x_{23},x_{31})=0,\quad
 P^{(7)}(x_1,x_2,x_3,x_{123})=0.
\end{equation}
\end{description} 
Since these equations relate the vertices of the quadrilateral on a lattice,
they are often called quad-equations or lattice equations.

%%%%%%%%%%%%%%%%%%%%%%%%%%%%%%
%% Figure 1
%%%%%%%%%%%%%%%%%%%%%%%%%%%%%%
\begin{figure}[t]
\begin{center}
\includegraphics[width=0.6\textwidth]{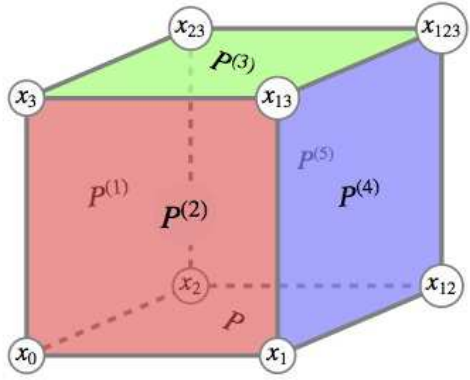}
\end{center}
\caption{The cube $C$ with the variables $x_0$, \dots, $x_{123}$ and face equations $P=0$ and $P^{(i)}=0$, $i=1,\dots,5$. 
Bottom: $P$,
Left: $P^{(1)}$,
Front: $P^{(2)}$,
Top: $P^{(3)}$,
Right: $P^{(4)}$,
Back: $P^{(5)}$.
}
\label{fig:cube}
\end{figure}
%%%%%%%%%%%%%%%%%%%%%%%%%%%%%%
%%%%%%%%%%%%%%%%%%%%%%%%%%%%%%

Some polynomials of ABS type are
\begin{align*}
 \text{Q1}&:Q1(x_1,x_2,x_3,x_4;\alpha_1,\alpha_2;\epsilon)\\
 &\quad=\alpha_1(x_1x_2+x_3x_4)-\alpha_2(x_1x_4+x_2x_3)
 -(\alpha_1-\alpha_2)(x_1x_3+x_2x_4)+\epsilon\alpha_1\alpha_2(\alpha_1-\alpha_2),\\
 \text{H3}&:H3(x_1,x_2,x_3,x_4;\alpha_1,\alpha_2;\delta;\epsilon)\\
 &\quad=\alpha_1(x_1x_2+x_3x_4)-\alpha_2(x_1x_4+x_2x_3)
 +({\alpha_1}^2-{\alpha_2}^2)\left(\delta+\cfrac{\epsilon}{\alpha_1\alpha_2}\,x_2x_4\right),\\
 \text{H1}&:H1(x_1,x_2,x_3,x_4;\alpha_1,\alpha_2;\epsilon)
 =(x_1-x_3)(x_2-x_4)+(\alpha_2-\alpha_1)(1-\epsilon x_2x_4),\\
\end{align*}
where $\alpha_1,\alpha_2\in\mathbb{C}^\ast$ and $\epsilon,\delta\in\{0,1\}$.
Many well known integrable partial difference equations (\PDE s) arise from assigning a polynomial of ABS type to quadrilaterals in the integer lattice $\bbZ^2$,
for example:
\begin{description}
\item[discrete Schwarzian KdV equation\cite{NC1995:MR1329559,NCWQ1984:MR763123}]
\begin{equation}\label{eqn:intro_DSKdV_1}
 Q1(U,\overline{U},\widehat{\overline{U}},\widehat{U};\alpha,\beta;0)=0
 ~\Leftrightarrow~
 \cfrac{(U-\overline{U})(\widehat{U}-\widehat{\overline{U}})}{(U-\widehat{U})(\overline{U}-\widehat{\overline{U}})}=\cfrac{\alpha}{\beta}\,;
\end{equation}
\item[lattice modified KdV equation\cite{NC1995:MR1329559,NQC1983:MR719638,ABS2003:MR1962121}]
\begin{equation}\label{eqn:intro_LMKdV_1}
 H3(U,\overline{U},-\widehat{\overline{U}},\widehat{U};\alpha,\beta;0;0)=0
 ~\Leftrightarrow~
 \cfrac{\widehat{\overline{U}}}{U}=\cfrac{\alpha\overline{U}-\beta\widehat{U}}{\alpha\widehat{U}-\beta\overline{U}}\,;
\end{equation}
\item[lattice potential KdV equation\cite{HirotaR1977:MR0460934,NC1995:MR1329559}]
\begin{equation}
 H1(U,\overline{U},\widehat{\overline{U}},\widehat{U};\alpha,\beta;0)=0
 ~\Leftrightarrow~
 (U-\widehat{\overline{U}})(\overline{U}-\widehat{U})=\alpha-\beta\,,
\end{equation}
\end{description}
where 
\begin{equation}\label{eqn:intro_notation_1}
 U=U_{l,m},\quad
 \alpha=\alpha_l,\quad
 \beta=\beta_m,\quad
 \bar{ }:l\to l+1,\quad
 \hat{ }:m\to m+1,\quad
 l,m\in\mathbb{Z}.
\end{equation}
Throughout this paper, we refer to  such \PDE s as ABS equations.

%%
%% from ABS eqn to Painleve eqn (dPへのリダクションはあるがPへの直接のリダクションはよくわかっていない)
%% ABS変数はPのハミルトニアンをあたえることから大切であろう．また，ルートやウェイト格子の性質ではなくABS方程式の背後にある整数格子がワイル群やラックス形式との相性が良い
%%
We note that in general a hypercube is said to be multi-dimensionally consistent, if all cubes contained in the hypercube are 3D consistent (see property {\bf(2)} above).
In a similar manner to the construction of ABS equations a hypercube causes a system of ABS equations by tilling it in the integer lattice (see, for example \S \ref{section:reduction_ABS}).

%%
%% Painleve and ABS equation
%%
It is well known that the Painlev\'e equations arise as the monodromy-preserving deformation of linear differential equations (see e.g., \cite{SHC2006:MR2263633,KNY2015:arXiv150908186K,book_NoumiM2004:MR2044201,FN1980:MR588248,FIKN:MR2264522} and reference therein).
The pair of linear differential equation and its deformation equation is referred to as the Lax representation (or, Lax pair) of the corresponding Painlev\'e equation.
It has also been reported that a compatibility condition of linear difference equation and its deformation equation also give a Painlev\'e equation \cite{KNY2015:arXiv150908186K,OR2016:arXiv160304393,AdlerVE1994:MR1280883}.
We here denote such a Lax representation as a {\it difference-differential Lax representation}.
Lax representations of the Painlev\'e equations usually arise by reductions from the integrable partial differential equations such as KdV equation, modified KdV equation, and so on.
In this paper, we show that difference-differential Lax representations of the Painlev\'e equations can be obtained from a system of integrable \PDE s of ABS type through periodic type reductions by using P$_{\rm IV}$ and P$_{\rm V}$ as examples.
Note that a Lax representation of an ABS equation is given by a pair of linear difference equation and its spectrum-preserving deformation. 
For a relation between monodromy- and spectrum- preserving deformations, we refer to \cite{FN1980:MR588248}.
%%%%%%%%%%%%%%%%%%%%%%%%%%%%%%%%%%%%%%%%%%%%%%%%%%%%%%%%%%%
%% 1.4 Plan of the paper
%%%%%%%%%%%%%%%%%%%%%%%%%%%%%%%%%%%%%%%%%%%%%%%%%%%%%%%%%%%
\subsection{Plan of the paper}
This paper is organized as follows.
In \S \ref{section:reduction_ABS}, we first define the system of ABS equations \eqref{eqn:standard_H1} and construct its Lax representation.
Then, we consider the reduction of system \eqref{eqn:standard_H1} to the system of ODEs \eqref{eqn:differential_omega}.
In \S \ref{section:affineWeylgroup_Lax}, using the symmetry of the integer lattice 
we obtain the affine Weyl group symmetry and the difference-differential Lax representation of system \eqref{eqn:differential_omega}.
In \S \ref{section:relation_Nobu_NY}, considering the relation between system \eqref{eqn:differential_omega} and NY-system,
we give the proofs of Theorems \ref{theorem:p4_hamiltonian}--\ref{theorem:p5_Lax}.
Some concluding remarks are given in \S \ref{ConcludingRemarks}.
%%%%%%%%%%%%%%%%%%%%%%%%%%%%%%%%%%%%%%%%%%%%%%%%%%%%%%%%%%%
%% 2. Reduction of a system of ABS equations to a system of ODEs
%%%%%%%%%%%%%%%%%%%%%%%%%%%%%%%%%%%%%%%%%%%%%%%%%%%%%%%%%%%
\section{Reduction of a system of ABS equations to a system of ODEs}
\label{section:reduction_ABS}
In this section, we consider the periodic reduction of the system of ABS equations \eqref{eqn:standard_H1} to the system of ODEs \eqref{eqn:differential_omega}, which is equivalent to NY-system (see \S \ref{section:relation_Nobu_NY} for the details).

In the same way that the lattice $\bbZ^2$ can be constructed by tiling the plane with squares, we construct the lattice $\bbZ^{n+2}$, where $n\in\bbZ_{>0}$, by tiling it with $(n+2)$-dimensional hypercubes (i.e. $(n+2)$-cubes).
We obtain a system of \PDE s on the lattice $\bbZ^{(n+2)}$ in a similar manner to the constructions of the ABS equations (see \S \ref{subsection:background}).
Indeed, assigning the function $u$ and H1$_{\ep=0}$ equations to the vertices and faces of each $(n+2)$-cube,
we obtain the following system of ABS equations:
\begin{equation}\label{eqn:standard_H1}
 (u-u_{\bi\,\bj})(u_\bi-u_\bj)+\al^{(i)}(l_i)-\al^{(j)}(l_j)=0,\quad 0\leq i<j\leq n+1,
\end{equation}
where $u=u(l_0,\dots,l_{n+1})$ is the dependent variable and 
$$\{\dots,\al^{(i)}(-1),\al^{(i)}(0),\al^{(i)}(1),\al^{(i)}(2),\dots\},\quad i=0,\dots,n+1,$$ 
are complex parameters.
Here, the subscript \,$\bi$ (or, $\ul{i}$\,) for an arbitrary function $F=F(l_0,\dots,l_{n+1})$
means \,$+1$ shift (or, $-1$ shift) in the $l_i$-direction, that is,
\begin{equation}
 F_{\bi}=F|_{l_i\mapsto l_i+1},\quad
 F_{\ul{i}}=F|_{l_i\mapsto l_i-1}.
\end{equation}
Below, we also use these notations for other objects.
For example, (2.1)$_\bi\,$ denotes (2.1)$|_{l_i\mapsto l_i+1}$.

We first rewrite system \eqref{eqn:standard_H1} by perceiving the $l_0$-direction as special.
Let $\al^{(0)}(l_0)$ and $u(l_0,\dots,l_{n+1})$ be the functions of  $t\in\bbC$ as follows:
\begin{equation}\label{eqn:l0_special_direction}
 \al^{(0)}(l_0)=\be(t+l_0\ep),\quad
 u(l_0,\dots,l_{n+1})=u_{l_1,\dots,l_{n+1}}(t+l_0\ep),
\end{equation} 
where $\ep\in\bbC$.
Then, system \eqref{eqn:standard_H1} can be rewritten as
\begin{subequations}\label{eqns:lpKdVs_u}
\begin{align}
 &(u-\bar{u}_\bj)(\bar{u}-u_\bj)+\be(t)-\al^{(j)}(l_j)=0,
 &&j=1,\dots, n+1,\label{eqn:lpKdVs_u_1}\\
 &(u-u_{\bi\,\bj})(u_\bi-u_\bj)+\al^{(i)}(l_i)-\al^{(j)}(l_j)=0,
 &&1\leq i<j\leq n+1,\label{eqn:lpKdVs_u_2}
\end{align}
\end{subequations}
where $u=u_{l_1,\dots,l_{n+1}}(t)$.
Here, the overline $~\bar{}~$ for an arbitrary function $F=F(t)$
means $+\ep$ shift of $t$, that is,
\begin{equation}
 \ol{F}=F(t+\ep).
\end{equation}
Note that system \eqref{eqns:lpKdVs_u} is not a special case of system \eqref{eqn:standard_H1}.
Indeed, shifting $t$ to $t+l_0\ep$ and using replacement \eqref{eqn:l0_special_direction},
we inversely obtain system \eqref{eqn:standard_H1} from system \eqref{eqns:lpKdVs_u}.

Following the method given in \cite{BS2002:MR1890049,NijhoffFW2002:MR1912127,WalkerAJ:thesis,JNS:paper3},
we obtain the Lax representation of system \eqref{eqns:lpKdVs_u}.
%%%%%%%%%%%%%%%%%%%%%%%%%%%%%%
%% Lemma 2.1
%%%%%%%%%%%%%%%%%%%%%%%%%%%%%%
\begin{lemma}
\label{lemma:Lax_n+2system}
The Lax representation of system \eqref{eqns:lpKdVs_u} is given by
\begin{subequations}\label{eqns:lax_psi}
\begin{align}
 &\ol{\psi}
 =\begin{pmatrix}1&u\\0&1\end{pmatrix}
 \begin{pmatrix}0&\mu-\be(t)\\1&-\bar{u}\end{pmatrix}
\psi,
\label{eqn:lax_psi_t}\\
 &\psi_\bi
 =\begin{pmatrix}1&u\\0&1\end{pmatrix}
 \begin{pmatrix}0&\mu-\al^{(i)}(l_i)\\1&-u_\bi\end{pmatrix}
\psi,\quad
 i=1,\dots,n+1,
\label{eqn:lax_psi_li}
\end{align}
\end{subequations}
where $u=u_{l_1,\dots,l_{n+1}}(t)$, $\psi=\psi_{l_1,\dots,l_{n+1}}(t)$ and $\mu\in\bbC$ is the spectral parameter,
that is, the compatibility conditions 
\begin{subequations}
\begin{align}
 &\ol{(\psi_\bj)}=\left(\,\ol{\psi}\,\right)_\bj,\quad j=1,\dots, n+1,\\
 &(\psi_\bi)_\bj=(\psi_\bj)_\bi,\quad 1\leq i<j\leq n+1,
\end{align}
\end{subequations}
are equivalent to \eqref{eqn:lpKdVs_u_1} and \eqref{eqn:lpKdVs_u_2}, respectively.
\end{lemma}
%%%%%%%%%%%%%%%%%%%%%%%%%%%%%%
%%%%%%%%%%%%%%%%%%%%%%%%%%%%%%

In Appendix \ref{section:appendix_proof_lemma_Lax}, 
we give the proof of Lemma \ref{lemma:Lax_n+2system}
and show how to construct a Lax representation of a system of ABS equations by using system \eqref{eqns:lpKdVs_u} as an example.

We next consider a periodic reduction of system \eqref{eqns:lpKdVs_u} and its Lax representation \eqref{eqns:lax_psi}.

%%%%%%%%%%%%%%%%%%%%%%%%%%%%%%
%% Lemma 2.2
%%%%%%%%%%%%%%%%%%%%%%%%%%%%%%
\begin{lemma}
\label{lemma:reduction}
Let
\begin{subequations}\label{eqns:def_U_phi}
\begin{align}
 &\be(t)=-\cfrac{t^2+2}{4}+\ep^{-2},\quad
 u_{l_1,\dots,l_{n+1}}(t)=U_{l_1,\dots,l_{n+1}}(t)+\left(\tilde{\al}_{l_1,\dots,l_{n+1}}-\ep^{-2}\right)t,\label{eqn:def_U}\\
 &\psi_{l_1,\dots,l_{n+1}}(t)=\ep^{-t/\ep}\begin{pmatrix}1&-\ep^{-2}t\\0&1\end{pmatrix}\phi_{l_1,\dots,l_{n+1}}(t),
\end{align}
\end{subequations}
where 
\begin{equation}
 \tilde{\al}_{l_1,\dots,l_{n+1}}=\sum_{i=1}^{n+1}\cfrac{\al^{(i)}(l_i)}{2(n+1)}.
\end{equation}
By imposing the $(1,\dots,1)$-periodic condition
\begin{equation}\label{eqn:periodc_condition}
 U_{l_1+1,\dots,l_{n+1}+1}(t)=U_{l_1,\dots,l_{n+1}}(t),
\end{equation}
with the following conditions of the parameters for $l\in\bbZ$$:$
\begin{equation}\label{eqn:para_al_be}
 \al^{(i)}(l)=\al^{(i)}(0)+(n+1)l,\quad i=1,\dots, n+1,
\end{equation}
system \eqref{eqns:lpKdVs_u} is reduced to the following system of equations:
\begin{subequations}\label{eqns:differential_difference_U}
\begin{align}
 &U'+U_\bj'-(U-U_\bj)(t-U+U_\bj)+2\tilde{\al}_{l_1,\dots,l_{n+1}}-\al^{(j)}(l_j)=0,
 && j=1,\dots, n+1,
 \label{eqn:differential_U}\\
 &(U-U_{\bi,\bj}-t)(U_\bi-U_\bj)+\al^{(i)}(l_i)-\al^{(j)}(l_j)=0,
 &&1\leq i<j\leq n+1,
 \label{eqn:difference_U}
\end{align}
\end{subequations}
where $U=U_{l_1,\dots,l_{n+1}}(t)$ and $'$ denotes ${\rm d}/{\rm d}t$.
Moreover, the Lax representation \eqref{eqns:lax_psi} is also reduced to the Lax representation of system \eqref{eqns:differential_difference_U} given by
\begin{subequations}\label{eqns:difference_differential_phi}
\begin{align}
 &\phi'
 =\begin{pmatrix}1&U+\tilde{\al}_{l_1,\dots,l_{n+1}}t\\0&1\end{pmatrix}
 \begin{pmatrix}0&-U'-\tilde{\al}_{l_1,\dots,l_{n+1}}+\frac{t^2+2}{4}+\mu\\1&-U-\tilde{\al}_{l_1,\dots,l_{n+1}}t\end{pmatrix}
 \phi,
 \label{eqn:differential_phi}\\
 &\phi_\bi
 =\begin{pmatrix}1&U+\tilde{\al}_{l_1,\dots,l_{n+1}}t\\0&1\end{pmatrix}
 \begin{pmatrix}0&-\al^{(i)}(l_i)+\mu\\1&-U_\bi-(\tilde{\al}_{l_1,\dots,l_{n+1}})_\bi\,t\end{pmatrix}
 \phi,\quad
 i=1,\dots,n+1,
 \label{eqn:difference_phi}
\end{align}
\end{subequations}
where $U=U_{l_1,\dots,l_{n+1}}(t)$ and $\phi=\phi_{l_1,\dots,l_{n+1}}(t)$, that is, 
the compatibility conditions 
\begin{subequations}
\begin{align}
 &\frac{{\rm d}}{{\rm d}t}(\phi_\bj)=(\phi')_\bj,&&\hspace{-8em} j=1,\dots, n+1,\\
 &(\phi_\bi)_\bj=(\phi_\bj)_\bi,&&\hspace{-8em} 1\leq i<j\leq n+1,
\end{align}
\end{subequations}
are equivalent to \eqref{eqn:differential_U} and \eqref{eqn:difference_U}, respectively.
\end{lemma}
%%%%%%%%%%%%%%%%%%%%%%%%%%%%%%
%%%%%%%%%%%%%%%%%%%%%%%%%%%%%%

The proof of Lemma \ref{lemma:reduction} is given in Appendix \ref{section:appendix_proof_lemma_reduction}.

We are now in a position to get a system of ODEs.
Let us define the variables $\omega_i=\omega_i(t)$ and the parameters $a_i$, $i=0,\dots,n$, by 
\begin{subequations}\label{eqns:def_omega_a}
\begin{align}
 &\omega_0=U_{0,\dots,0}(t),\quad
 \omega_1=U_{1,0,\dots,0}(t),\quad \dots,\quad
 \omega_n=U_{1,\dots,1,0}(t),\\
 &a_0=\cfrac{\al^{(1)}(0)-\al^{(n+1)}(0)}{n+1}+1,\quad
 a_i=\cfrac{\al^{(i+1)}(0)-\al^{(i)}(0)}{n+1},\quad i=1,\dots,n,
\end{align}
\end{subequations}
where
\begin{equation}\label{eqn:sum_ai}
 \sum_{i=0}^na_i=1.
\end{equation}
Substituting 
\begin{equation}
 l_1=\cdots=l_i=1,\quad l_{i+1}=\cdots=l_{n+1}=0,\quad j=i+1,
\end{equation}
into system \eqref{eqn:differential_U} and using the relation
\begin{equation}
 2\tilde{\al}_{0,\dots,0}-\al^{(1)}(0)
% =&\sum_{k=1}^{n+1}\cfrac{\al^{(k)}(0)}{n+1}-\al^{(1)}(0)
% =\sum_{k=1}^{n+1}\cfrac{(n+2-k)\left(\al^{(k+1)}(0)-\al^{(k)}(0)\right)}{n+1}\notag\\
% =&\sum_{k=1}^{n}(n+2-k)a_k+a_0-1
 =\sum_{k=1}^{n}(n+1-k)a_k,
\end{equation}
which can be verified by the direct calculation,
we obtain the following system of ODEs:
\begin{equation}\label{eqn:differential_omega}
 \omega_i'+\omega_{i+1}'=-(\omega_i-\omega_{i+1})(\omega_i-\omega_{i+1}-t)-\sum_{k=1}^n (n+1-k)a_{i+k},
\end{equation}
where $i\in \bbZ/(n+1)\bbZ$.
System \eqref{eqn:differential_omega} is equivalent to NY-system (see \S \ref{section:relation_Nobu_NY} for the details).
Before explaining it, in the next section we consider the affine Weyl group symmetry and Lax representation of system \eqref{eqn:differential_omega}.

%%%%%%%%%%%%%%%%%%%%%%%%%%%%%%%%%%%%%%%%%%%%%%%%%%%%%%%%%%%
%% 3. Affine Weyl group Symmetry and Lax representation of system \eqref{eqn:differential_omega}
%%%%%%%%%%%%%%%%%%%%%%%%%%%%%%%%%%%%%%%%%%%%%%%%%%%%%%%%%%%
\section{Affine Weyl group Symmetry and Lax representation of system \eqref{eqn:differential_omega}}
\label{section:affineWeylgroup_Lax}
In this section, we first consider a linear action of an affine Weyl group of $A$-type, which is the symmetry of the integer lattice.
Then, lifting its action to a birational action we obtain affine Weyl group symmetry of system \eqref{eqn:differential_omega}.
Using the symmetry group together with Lemma \ref{lemma:Lax_n+2system}, we finally construct the difference-differential Lax representation of system \eqref{eqn:differential_omega}.

%%%%%%%%%%%%%%%%%%%%%%%%%%%%%%%%%%%%%%%%%%%%%%%%%%%%%%%%%%%
%% 3.1. Affine Weyl group Symmetry of the lattice $\bbZ^{n+1}$
%%%%%%%%%%%%%%%%%%%%%%%%%%%%%%%%%%%%%%%%%%%%%%%%%%%%%%%%%%%
\subsection{Affine Weyl group Symmetry of the lattice $\bbZ^{n+1}$}
\label{subsection:affine_linear}
In this section, considering a symmetry of the integer lattice $\bbZ^{n+1}$, we obtain the linear action of the affine Weyl group $\widetilde{W}(A_n^{(1)})$.

We define the automorphisms of the lattice $\bbZ^{n+1}$: $s_i$, $i=0,\dots,n$, and $\pi$ 
by the following actions on the coordinates $(l_1,\dots,l_{n+1})\in\bbZ^{n+1}$:
\begin{subequations}\label{eqns:WAn_lattice}
\begin{align}
 &s_0:(l_1,\dots,l_{n+1})\mapsto (l_{n+1}+1,l_2,\dots,l_n,l_1-1),\\
 &s_i:(l_1\dots l_{n+1})\mapsto (l_1,\dots,l_{i-1},l_{i+1},l_i,l_{i+2},\dots, l_{n+1}),\quad i=1,\dots,n,\\
 &\pi:(l_1,\dots,l_{n+1})\mapsto (l_{n+1}+1,l_1,\dots,l_n).
\end{align}
\end{subequations}
For convenience, throughout this paper we use the following notation for the combined transformation of arbitrary mappings $w$ and $w'$:
\begin{equation}
 ww':=w\circ w'.
\end{equation}
The group of automorphisms $\langle s_0,\dots,s_n,\pi\rangle$ forms the extended affine Weyl group of type $A_n^{(1)}$,
denoted by $\widetilde{W}(A_n^{(1)})$.
Namely, they satisfy the following fundamental relations:
\begin{equation}\label{eqn:fundamental_An}
 {s_i}^2=1,\quad 
 (s_is_{i\pm 1})^3=1,\quad
 (s_is_j)^2=1,\quad j\neq i,i\pm 1,\quad
 \pi s_i=s_{i+1}\pi,
\end{equation}
where $i,j\in \bbZ/(n+1)\bbZ$.
Note that $\widetilde{W}(A_n^{(1)})$ is not the ``full" extended affine Weyl group of type $A_n^{(1)}$, 
since it only includes rotational symmetries of the affine Dynkin diagram, but not the reflections.

%%%%%%%%%%%%%%%%%%%%%%%%%%%%%%
%% Remark 3.1
%%%%%%%%%%%%%%%%%%%%%%%%%%%%%%
\begin{remark}
Action of each element of $\widetilde{W}(A_n^{(1)})$ on the coordinates $(l_1,\dots,l_{n+1})\in\bbZ^{n+1}$ is defined, 
but that on the each lattice parameter $l_i$ is not defined.
For example, in the case $n=2$ the transformation $\pi$ acts on the origin as the following (see Figure \ref{fig:fig_3Dlattice}):
\begin{equation}
 \pi.(0,0,0)=(1,0,0),
\end{equation}
but it cannot act on the parameter $l_i$ like $\pi.l_i$.
\end{remark}
%%%%%%%%%%%%%%%%%%%%%%%%%%%%%%
%%%%%%%%%%%%%%%%%%%%%%%%%%%%%%

In the lattice $\bbZ^{n+1}$, there are $(n+1)$ orthogonal directions, which naturally give rise to $(n+1)$ translation operators.
Operators $T_i$, $i=1,\dots,n+1$, whose actions on the coordinates $(l_1,\dots,l_{n+1})\in\bbZ^{n+1}$ are given by
\begin{equation}
 T_i:(l_1,\dots,l_{n+1})\mapsto (l_1,\dots,l_{n+1})_{\bi}\,,
\end{equation}
can be expressed by the elements of $\widetilde{W}(A_n^{(1)})$ as the following:
\begin{equation}
 T_i=\pi s_{i+n-1}\cdots s_{i+1}s_i,
\end{equation}
where $i\in \bbZ/(n+1)\bbZ$.
Note that $\pi^{n+1}$ is also a translation operator whose action is given by
\begin{equation}\label{eqn:pi_nonperiod}
 \pi^{n+1}:(l_1,\dots,l_{n+1})\mapsto (l_1+1,\dots,l_{n+1}+1),
\end{equation}
and can be expressed by compositions of $T_i$ as the following:
\begin{equation}
 \pi^{n+1}=T_1\cdots T_{n+1}.
\end{equation}

%%%%%%%%%%%%%%%%%%%%%%%%%%%%%%
%% Figure 2
%%%%%%%%%%%%%%%%%%%%%%%%%%%%%%
\begin{figure}[t]
\begin{center}
\includegraphics[width=0.9\textwidth]{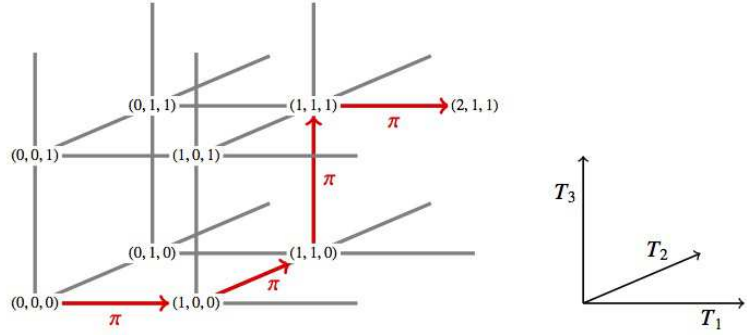}
\end{center}
\caption{Case $n=2$. The actions of $T_i$, $i=1,2,3$, on the lattice $\bbZ^3$ and the action of $\pi$ around the origin are described.
}
\label{fig:fig_3Dlattice}
\end{figure}
%%%%%%%%%%%%%%%%%%%%%%%%%%%%%%
%%%%%%%%%%%%%%%%%%%%%%%%%%%%%%

%%%%%%%%%%%%%%%%%%%%%%%%%%%%%%%%%%%%%%%%%%%%%%%%%%%%%%%%%%%
%% 3.2. Affine Weyl group Symmetry of system \eqref{eqn:differential_omega}
%%%%%%%%%%%%%%%%%%%%%%%%%%%%%%%%%%%%%%%%%%%%%%%%%%%%%%%%%%%
\subsection{Affine Weyl group Symmetry of system \eqref{eqn:differential_omega}}
In this section, extending the linear action of $\widetilde{W}(A_n^{(1)})$ given in \S \ref{subsection:affine_linear} to the birational action, 
we obtain the affine Weyl group symmetry of system \eqref{eqn:differential_omega}.

From the definition, the variables $U_{l_1,\dots,l_{n+1}}(t)$, defined by \eqref{eqn:def_U}, are assigned on the vertices
and the quad-equations \eqref{eqn:difference_U} are assigned on the faces of the lattice $\bbZ^{n+1}$.
Therefore, we can naturally lift the action of $\widetilde{W}(A_n^{(1)})$ to the actions on the function $U_{l_1,\dots,l_{n+1}}(t)$ by
\begin{subequations}\label{eqns:birational_An}
\begin{equation}
 s_i.U_{(l_1,\dots,l_{n+1})}=U_{s_i.(l_1,\dots,l_{n+1})},\quad i=0,\dots,n,\qquad
 \pi.U_{(l_1,\dots,l_{n+1})}=U_{\pi.(l_1,\dots,l_{n+1})},
\end{equation}
where  
\begin{equation}
 U_{(l_1,\dots,l_{n+1})}=U_{l_1,\dots,l_{n+1}}(t).
\end{equation}
Moreover, we define the action of $\widetilde{W}(A_n^{(1)})$ on the parameters $\al^{(i)}(l)$, $i=1,\dots,n+1$, where $l\in\bbZ$\,, as the following:
\begin{align}
 &s_0.\al^{(j)}(l)
 =\begin{cases}
 \al^{(n+1)}(l-1)&\text{if\hspace{1em}}j=1,\\
 \al^{(1)}(l+1)&\text{if\hspace{1em}}j=n+1,\\
 \al^{(j)}(l)&\text{otherwise},
 \end{cases}\\
 &s_i.\al^{(j)}(l)
 =\begin{cases}
 \al^{(j+1)}(l)&\text{if\hspace{1em}}j=i,\\
 \al^{(j-1)}(l)&\text{if\hspace{1em}}j=i+1,\\
 \al^{(j)}(l)&\text{otherwise},
 \end{cases}\\
 &\pi.\al^{(j)}(l)
 =\begin{cases}
 \al^{(j+1)}(l)&\text{if\hspace{1em}}j=1,\dots,n,\\
 \al^{(1)}(l+1)&\text{if\hspace{1em}}j=n+1,
 \end{cases}
\end{align}
\end{subequations}
where $i=1,\dots,n$.
These actions give the actions on the parameters $a_i$ and the variables $\omega_i$ as the following lemma.
%%%%%%%%%%%%%%%%%%%%%%%%%%%%%%
%% Lemma
%%%%%%%%%%%%%%%%%%%%%%%%%%%%%%
\begin{lemma}\label{lemma:action_omega_a}
The actions of $\widetilde{W}(A_n^{(1)})=\langle s_0,\dots,s_n,\pi\rangle$ 
on the parameters $a_i$, $i=0,\dots,n$, are given by
\begin{equation}\label{eqn:affineWeylaction_para}
 s_i(a_j)
 =\begin{cases}
 -a_j&\text{if\hspace{1em}}j=i,\\
 a_j+ a_i&\text{if\hspace{1em}}j=i\pm 1,\\
 a_j&\text{otherwise},
 \end{cases}\qquad
 \pi(a_i)=a_{i+1},
\end{equation}
where $i,j\in \bbZ/(n+1)\bbZ$,
while those on the variables $\omega_i$\,, $i=0,\dots,n$, are given by
\begin{subequations}
\begin{align}
 &s_i(\omega_j)
 =\begin{cases}
 \omega_i+\cfrac{(n+1)a_i}{t-\omega_{i+n}+\omega_{i+1}}&\text{if\hspace{1em}}j=i,\\
 \omega_j&\text{if\hspace{1em}}j\neq i,
 \end{cases}
 \label{eqn:affineWeylaction_omega_s}\\
 &\pi (\omega_i)=\omega_{i+1},
\end{align}
\end{subequations}
where $i,j\in \bbZ/(n+1)\bbZ$.
\end{lemma}
%%%%%%%%%%%%%%%%%%%%%%%%%%%%%%
%% proof
%%%%%%%%%%%%%%%%%%%%%%%%%%%%%%
\begin{proof}
From the periodic condition \eqref{eqn:periodc_condition}, the definition \eqref{eqns:def_omega_a} and the actions \eqref{eqns:birational_An}, we obtain
\begin{align}
 &s_0(\omega_0)=U_{(1,0,\dots,0,-1)}=U_{(2,1,\dots,1,0)}=T_1(\omega_n),\label{eqn:proof_s0omega0}\\
 &s_k(\omega_k)=U_{(1,\dots,1,0,1,0,\dots,0)}=T_{k+1}(\omega_{k-1}),\quad k=1,\dots,n.\label{eqn:proof_skomegak}
\end{align}
Moreover, substituting 
\begin{align}
 & l_1=\cdots=l_n=1,\quad l_{n+1}=0,\quad i=n+1,\quad j=1,\\
 & l_1=\cdots=l_{k-1}=1,\quad l_k=\cdots=l_{n+1}=0,\quad i=k,\quad j=k+1
\end{align}
into \eqref{eqn:difference_U}, we obtain
\begin{align}
 &T_1(\omega_n)
 =\omega_0+\cfrac{-\al^{(n+1)}(0)+\al^{(1)}(1)}{t-\omega_n+\omega_1}
 =\omega_0+\cfrac{(n+1)a_0}{t-\omega_n+\omega_1},\label{eqn:proof_T1omegan}\\
 &T_{k+1}(\omega_{k-1})
 =\omega_k+\cfrac{-\al^{(k)}(0)+\al^{(k+1)}(0)}{t-\omega_{k-1}+\omega_{k+1}}
 =\omega_k+\cfrac{(n+1)a_k}{t-\omega_{k+n}+\omega_{k+1}},\label{eqn:proof_Tk+1omegak-1}
\end{align}
where $k=1,\dots,n$, respectively.
Therefore, from Equations \eqref{eqn:proof_s0omega0}, \eqref{eqn:proof_skomegak}, \eqref{eqn:proof_T1omegan} and \eqref{eqn:proof_Tk+1omegak-1},
the actions \eqref{eqn:affineWeylaction_omega_s} hold.
From the actions \eqref{eqns:birational_An} and the definition \eqref{eqns:def_omega_a}, the others can be easily verified.
Therefore, we have completed the proof.
\end{proof}
%%%%%%%%%%%%%%%%%%%%%%%%%%%%%%
%%%%%%%%%%%%%%%%%%%%%%%%%%%%%%

In general, for a function $F=F(a_i,\omega_j)$, we let an element $w\in\widetilde{W}(A_n^{(1)})$
act as $w.F=F(w.a_i,w.\omega_j)$, that is, $w$ acts on the arguments from the left.
We can easily verify that under the birational actions given in Lemma \ref{lemma:action_omega_a},
$\widetilde{W}(A_n^{(1)})$ satisfies the fundamental relations \eqref{eqn:fundamental_An} and the following relation: 
\begin{equation}\label{eqn:pi_period}
 \pi^{n+1}=1.
\end{equation}
%Note that under the actions on the lattice $\bbZ^{n+1}$, the relation \eqref{eqn:pi_period} does not hold as shown in \eqref{eqn:pi_nonperiod}.
We can also verify that by using system \eqref{eqn:differential_omega} and the birational actions of $\widetilde{W}(A_n^{(1)})$ given in Lemma \ref{lemma:action_omega_a}, the following relations hold:
\begin{equation}
 s_j(\omega_i'+\omega_{i+1}')=\frac{{\rm d}}{{\rm d}t}\Big(s_j(\omega_i)+s_j(\omega_{i+1})\Big),\quad
 \pi(\omega_i'+\omega_{i+1}')=\frac{{\rm d}}{{\rm d}t}\Big(\pi(\omega_i)+\pi(\omega_{i+1})\Big),
\end{equation}
where $i,j\in \bbZ/(n+1)\bbZ$,
which indicate that $\widetilde{W}(A_n^{(1)})$ is a B\"acklund transformation group of system \eqref{eqn:differential_omega}.
Therefore, the following lemma holds.
%%%%%%%%%%%%%%%%%%%%%%%%%%%%%%
%% Lemma 3.3
%%%%%%%%%%%%%%%%%%%%%%%%%%%%%%
\begin{lemma}
B\"acklund transformations of system \eqref{eqn:differential_omega} collectively form 
the extended affine Weyl group $\widetilde{W}(A_n^{(1)})$.
\end{lemma}
%%%%%%%%%%%%%%%%%%%%%%%%%%%%%%
%%%%%%%%%%%%%%%%%%%%%%%%%%%%%%

%%%%%%%%%%%%%%%%%%%%%%%%%%%%%%%%%%%%%%%%%%%%%%%%%%%%%%%%%%%
%% 3.3. Difference-differential Lax representation of system \eqref{eqn:differential_omega}
%%%%%%%%%%%%%%%%%%%%%%%%%%%%%%%%%%%%%%%%%%%%%%%%%%%%%%%%%%%
\subsection{Difference-differential Lax representation of system \eqref{eqn:differential_omega}}
In this section, we define the birational action of $\widetilde{W}(A_n^{(1)})$ on the wave function $\phi$.
Using this action, we obtain the difference-differential Lax representation of system \eqref{eqn:differential_omega}.

In a similar manner to $U_{l_1,\dots,l_{n+1}}(t)$, 
we can assign $\phi_{l_1,\dots,l_{n+1}}(t)$ on the vertices $(l_1,\dots,l_{n+1})\in\bbZ^{n+1}$.
Then, the action of $\widetilde{W}(A_n^{(1)})$ can be lifted to the action on the function $\phi_{l_1,\dots,l_{n+1}}(t)$ as the following:
\begin{equation}\label{eqns:birational_An_phi}
 s_i.\phi_{(l_1,\dots,l_{n+1})}=\phi_{s_i.(l_1,\dots,l_{n+1})},\quad
 \pi.\phi_{(l_1,\dots,l_{n+1})}=\phi_{\pi.(l_1,\dots,l_{n+1})},
\end{equation}
where 
\begin{equation}
 \phi_{(l_1,\dots,l_{n+1})}=\phi_{l_1,\dots,l_{n+1}}(t).
\end{equation}
Let us define the variables $\Phi_i=\Phi_i(t)$, $i=0,\dots,n$, and the parameter $x$ by
\begin{equation}\label{eqns:def_phi_x}
 \Phi_0=\phi_{(0,\dots,0)},\quad
 \Phi_1=\phi_{(1,0,\dots,0)},\quad \dots,\quad
 \Phi_n=\phi_{(1,\dots,1,0)},\quad
 x=\cfrac{\al^{(1)}(0)}{n+1}.
\end{equation}
Substituting 
\begin{equation}
 l_1=\cdots=l_k=1,\quad l_{k+1}=\cdots,l_{n+1}=0
\end{equation}
into system \eqref{eqns:difference_differential_phi}, we obtain the following system of equations:
\begin{subequations}
\begin{align}
 &\Phi_k'
 =\begin{pmatrix}1&\omega_k+(\tilde{\al}+\frac{k}{2})t\\0&1\end{pmatrix}
 \begin{pmatrix}0&-\omega_k'-\tilde{\al}-\frac{k}{2}+\frac{t^2+2}{4}+\mu\\1&-\omega_k-(\tilde{\al}+\frac{k}{2})t\end{pmatrix}
 \Phi_k,
 \label{eqn:differential_phi_i}\\
 &T_i(\Phi_k)
 =\begin{cases}
 ~\begin{pmatrix}1&\omega_k+(\tilde{\al}+\frac{k}{2})t\\0&1\end{pmatrix}
 \begin{pmatrix}0&-\al^{(i)}-(n+1)+\mu\\1&-\left(\omega_k\right)_\bi-(\tilde{\al}+\frac{k+1}{2})t\end{pmatrix}
 \Phi_k
 &\text{if\hspace{1em}}i\leq k,\\[1.5em]
 ~\begin{pmatrix}1&\omega_k+(\tilde{\al}+\frac{k}{2})t\\0&1\end{pmatrix}
 \begin{pmatrix}0&-\al^{(i)}+\mu\\1&-\left(\omega_k\right)_\bi-(\tilde{\al}+\frac{k+1}{2})t\end{pmatrix}
 \Phi_k
 &\text{if\hspace{1em}}i>k,
 \end{cases}
 \label{eqn:differece_phi_i}
\end{align}
\end{subequations}
where $k=0,\dots,n$ and $\al^{(i)}$, $i=1,\dots,n+1$, and $\tilde{\al}$ are defined by
\begin{equation}\label{eqn:altal_ax}
 \al^{(i)}=\al^{(i)}(0)=(n+1)\left(x+\sum_{j=1}^{i-1}a_j\right),\quad
 \tilde{\al}=\tilde{\al}_{0,\dots,0}=\cfrac{(n+1)x}{2}+\sum_{j=1}^{n}\cfrac{(n+1-j)a_j}{2}.
\end{equation}
Then, the action of $\widetilde{W}(A_n^{(1)})$ on the variables $\Phi_i$, $i=0,\dots,n$, and the parameter $x$ are given in the following lemma.
%%%%%%%%%%%%%%%%%%%%%%%%%%%%%%
%% Lemma
%%%%%%%%%%%%%%%%%%%%%%%%%%%%%%
\begin{lemma}\label{lemma:action_phi_x}
The action of $\widetilde{W}(A_n^{(1)})=\langle s_0,\dots,s_n,\pi\rangle$ 
on the parameter $x$ is given by
\begin{equation}\label{eqn:affineWeylaction_x}
 s_i(x)
 =\begin{cases}
 x-a_0&\text{if\hspace{1em}}i=0,\\
 x+a_1&\text{if\hspace{1em}}i=1,\\
 x&\text{otherwise},
 \end{cases}\qquad
 \pi(x)=x+a_1,
\end{equation}
while that on the variables $\Phi_k$\,, $k=0,\dots,n$, is given by
\begin{subequations}\label{eqns:affineWeylaction_phi}
\begin{align}
 &s_0(\Phi_0)
 =\begin{pmatrix}0&-\al^{(n+1)}+n+1+\mu\\1&-\omega_1-(\tilde{\al}+\frac{1}{2})t\end{pmatrix}^{-1}
 \begin{pmatrix}1&\omega_0-\left(\omega_1\right)_{\ul{n+1}}\\0&1\end{pmatrix}
 \begin{pmatrix}0&-\al^{(1)}+\mu\\1&-\omega_1-(\tilde{\al}+\frac{1}{2})t\end{pmatrix}
 \Phi_0,
 \label{eqn:affineWeylaction_s0phi0}\\
 &s_i(\Phi_i)
 =\begin{pmatrix}1&\omega_{i-1}+(\tilde{\al}+\frac{i-1}{2})t\\0&1\end{pmatrix}
 \begin{pmatrix}\dfrac{\al^{(i+1)}-\mu}{\al^{(i)}-\mu}&0\\\dfrac{T_{i+1}\left(\omega_{i-1}\right)-\omega_i}{\al^{(i)}-\mu}&1\end{pmatrix}
 \begin{pmatrix}1&\omega_{i-1}+(\tilde{\al}+\frac{i-1}{2})t\\0&1\end{pmatrix}^{-1}
 \Phi_i,
 \label{eqn:affineWeylaction_siphii}\\
 &s_j(\Phi_k)=\Phi_k,\quad j,k=0,\dots,n,\quad j\neq k,\\
 &\pi(\Phi_k)
 =\begin{pmatrix}1&\omega_k+(\tilde{\al}+\frac{k}{2})t\\0&1\end{pmatrix}
 \begin{pmatrix}0&-\al^{(k+1)}+\mu\\1&-\omega_{k+1}-(\tilde{\al}+\frac{k+1}{2})t\end{pmatrix}
 \Phi_k,\quad
 k=0,\dots,n,
 \label{eqn:affineWeylaction_piphik}
\end{align}
\end{subequations}
where $i=1,\dots,n$ and $\omega_{n+1}=\omega_0$.
\end{lemma}
%%%%%%%%%%%%%%%%%%%%%%%%%%%%%%
%% proof
%%%%%%%%%%%%%%%%%%%%%%%%%%%%%%
%この証明は簡単に計算できるので，使ったマスマティカファイルは残していない
\begin{proof}
From \eqref{eqns:birational_An}, \eqref{eqns:birational_An_phi}, \eqref{eqn:differece_phi_i} and \eqref{eqns:def_phi_x}, we obtain
\begin{subequations}\label{eqns:proof_action_phi_x_1}
\begin{align}
 &s_0(\Phi_0)=\phi_{(1,0,\dots,0,-1)}(t)={T_{n+1}}^{-1}\left(\Phi_1\right),\\ 
 &T_{n+1}\left(\Phi_1\right)
 =\begin{pmatrix}1&\omega_1+(\tilde{\al}+\frac{1}{2})t\\0&1\end{pmatrix}
 \begin{pmatrix}0&-\al^{(n+1)}+\mu\\1&-T_{n+1}\left(\omega_1\right)-(\tilde{\al}+1)t\end{pmatrix}
 \Phi_1,\\
 &\Phi_1
 =T_1(\Phi_0)
 =\begin{pmatrix}1&\omega_0+\tilde{\al}t\\0&1\end{pmatrix}
 \begin{pmatrix}0&-\al^{(1)}+\mu\\1&-\omega_1-(\tilde{\al}+\frac{1}{2})t\end{pmatrix}
 \Phi_0,
\end{align}
\end{subequations}
\begin{subequations}\label{eqns:proof_action_phi_x_2}
\begin{align}
 &s_i(\Phi_i)=\phi_{(1,\dots,1,0,1,0,\dots,0)}(t)=T_{i+1}\left(\Phi_{i-1}\right),\\
 &T_{i+1}\left(\Phi_{i-1}\right)
 =\begin{pmatrix}1&\omega_{i-1}+(\tilde{\al}+\frac{i-1}{2})t\\0&1\end{pmatrix}
 \begin{pmatrix}0&-\al^{(i+1)}+\mu\\1&-T_{i+1}\left(\omega_{i-1}\right)-(\tilde{\al}+\frac{i}{2})t\end{pmatrix}
 \Phi_{i-1},\\
 &\Phi_i=T_i(\Phi_{i-1})
 =\begin{pmatrix}1&\omega_{i-1}+(\tilde{\al}+\frac{i-1}{2})t\\0&1\end{pmatrix}
 \begin{pmatrix}0&-\al^{(i)}+\mu\\1&-\omega_i-(\tilde{\al}+\frac{i}{2})t\end{pmatrix}
 \Phi_{i-1},
\end{align}
\end{subequations}
\begin{subequations}\label{eqns:proof_action_phi_x_3}
\begin{align}
 &\pi(\Phi_{k-1})=T_k\left(\Phi_{k-1}\right),\\
 &T_k\left(\Phi_{k-1}\right)
 =\begin{pmatrix}1&\omega_{k-1}+(\tilde{\al}+\frac{k-1}{2})t\\0&1\end{pmatrix}
 \begin{pmatrix}0&-\al^{(k)}+\mu\\1&-\omega_k-(\tilde{\al}+\frac{k}{2})t\end{pmatrix}
 \Phi_{k-1}.
\end{align}
\end{subequations}
Therefore, from \eqref{eqns:proof_action_phi_x_1}, \eqref{eqns:proof_action_phi_x_2} and \eqref{eqns:proof_action_phi_x_3}, we obtain \eqref{eqn:affineWeylaction_s0phi0}, \eqref{eqn:affineWeylaction_siphii} and \eqref{eqn:affineWeylaction_piphik}, respectively.
From the actions \eqref{eqns:WAn_lattice} and \eqref{eqns:birational_An_phi} and the definition \eqref{eqns:def_phi_x}, the others can be easily verified. 
Therefore, we have completed the proof.
\end{proof}
%%%%%%%%%%%%%%%%%%%%%%%%%%%%%%
%%%%%%%%%%%%%%%%%%%%%%%%%%%%%%

Note that we can easily verify that under the actions on the variables $\Phi_i$ and the parameter $x$,
$\widetilde{W}(A_n^{(1)})$ satisfies the fundamental relations \eqref{eqn:fundamental_An}
but does not satisfy the relation \eqref{eqn:pi_period}.
This unsatisfied relation is a key to construct a difference-differential Lax representation of an ODE.

We are now in a position to construct a Lax representation of system \eqref{eqn:differential_omega}.
Let us define the shift operator of $x$ by
\begin{equation}
 T_x=\pi^{n+1}.
\end{equation}
The action of $T_x$ on the parameter $x$ is given by
\begin{equation}
 T_x:x\mapsto x+1,
\end{equation}
while that on the variables $\omega_i$ and parameters $a_i$, $i=0,\dots,n$, is given by an identity mapping, i.e.,
\begin{equation}
 T_x(\omega_i)=\omega_i,\quad 
 T_x(a_i)=a_i,\quad 
 i=0,\dots,n.
\end{equation}
Therefore, from \eqref{eqn:differential_phi_i} and \eqref{eqn:affineWeylaction_piphik} we obtain the following lemma.
%%%%%%%%%%%%%%%%%%%%%%%%%%%%%%
%% Lemma
%%%%%%%%%%%%%%%%%%%%%%%%%%%%%%
\begin{lemma}\label{lemma:lax_Ui}
The difference-differential Lax representation of system \eqref{eqn:differential_omega} is given by the following:
\begin{align}
 &T_x(\Phi_i)=\pi^{n+1}(\Phi_i),\\
 &\Phi_i'
 =\begin{pmatrix}1&\omega_i+(\tilde{\al}+\frac{i}{2})t\\0&1\end{pmatrix}
 \begin{pmatrix}0&-\omega_i'-\tilde{\al}-\frac{i}{2}+\frac{t^2+2}{4}+\mu\\1&-\omega_i-(\tilde{\al}+\frac{i}{2})t\end{pmatrix}
 \Phi_i,
\end{align}
where 
\begin{align}
 &\pi(\Phi_i)
 =\begin{pmatrix}1&\omega_i+(\tilde{\al}+\frac{i}{2})t\\0&1\end{pmatrix}
 \begin{pmatrix}0&-\al^{(i+1)}+\mu\\1&-\omega_{i+1}-(\tilde{\al}+\frac{i+1}{2})t\end{pmatrix}
 \Phi_i,\\
 &\pi(\tilde{\al})=\tilde{\al}+\frac{1}{2},\quad
 \pi(t)=t,\quad
 \pi(\mu)=\mu,\quad
 \pi(\omega_i)=\omega_{i+1},\\
 &\pi(\al^{(i)})
 =\begin{cases}
 \al^{(i+1)}&\text{if\hspace{1em}}i=1,\dots,n,\\
 \al^{(1)}+n+1&\text{if\hspace{1em}}i=n+1,
 \end{cases}
\end{align}
that is, the compatibility conditions 
\begin{equation}
 \frac{{\rm d}}{{\rm d}t}T_x(\Phi_i)=T_x(\Phi_i'),\quad i=0,\dots,n,
\end{equation}
are equivalent to system \eqref{eqn:differential_omega}.
Note that the relations between the parameters $\al^{(i)}$, $\tilde{\al}$ and the parameters $a_j$, $x$ are given by \eqref{eqn:altal_ax}.
\end{lemma}
%%%%%%%%%%%%%%%%%%%%%%%%%%%%%%
%%%%%%%%%%%%%%%%%%%%%%%%%%%%%%

%%%%%%%%%%%%%%%%%%%%%%%%%%%%%%%%%%%%%%%%%%%%%%%%%%%%%%%%%%%
%% 4. Difference-differential Lax representations of P$_{\rm IV}$ and P$_{\rm V}$
%%%%%%%%%%%%%%%%%%%%%%%%%%%%%%%%%%%%%%%%%%%%%%%%%%%%%%%%%%%
\section{Difference-differential Lax representations of P$_{\rm IV}$ and P$_{\rm V}$}
\label{section:relation_Nobu_NY}
In this section, considering the relation between system \eqref{eqn:differential_omega} and NY-system,
we obtain the difference-differential Lax representations of P$_{\rm IV}$ \eqref{eqn:intro_P4} and P$_{\rm V}$ \eqref{eqn:intro_P5}.

Let us define the variables $g_i=g_i(t)$, $i=0,\dots,n$, by
\begin{equation}
 g_i=\omega_i-\omega_{i+1}-\frac{t}{2},
\end{equation}
where $\omega_{n+1}=\omega_0$.
Then, system \eqref{eqn:differential_omega} can be rewritten as the periodic dressing chain with period $(n+1)$\cite{VS1993:MR1251164}:
\begin{equation}\label{eqn:periodic_dressing_chain}
 g_{i+1}'+g_i'={g_{i+1}}^2-{g_i}^2-(n+1)a_{i+1},
\end{equation}
where $i\in \bbZ/(n+1)\bbZ$.
It is well known that from system \eqref{eqn:periodic_dressing_chain}, 
we can obtain NY-system containing P$_{\rm IV}$ and P$_{\rm V}$\cite{SHC2006:MR2263633,AdlerVE1994:MR1280883}.
Through this relation we construct the Lax representations of P$_{\rm IV}$ and P$_{\rm V}$
from the Lax representation of system \eqref{eqn:differential_omega}.
%%%%%%%%%%%%%%%%%%%%%%%%%%%%%%%%%%%%%%%%%%%%%%%%%%%%%%%%%%%
%% 4.1 Case $n=2$: the Painlev\'e IV equation
%%%%%%%%%%%%%%%%%%%%%%%%%%%%%%%%%%%%%%%%%%%%%%%%%%%%%%%%%%%
\subsection{Case $n=2$: the Painlev\'e IV equation}\label{subsection:P4}
In this section considering the case $n=2$ we obtain the difference-differential Lax representation of P$_{\rm IV}$ from that of system \eqref{eqn:differential_omega}.

Let us define the variables $f_i=f_i(t)$, $i=0,1,2$, by
\begin{equation}
 f_0=\omega_1-\omega_2+t,\quad
 f_1=\omega_2-\omega_0+t,\quad
 f_2=\omega_0-\omega_1+t.
\end{equation}
Then, from system \eqref{eqn:differential_omega} and the condition for the parameters \eqref{eqn:sum_ai} with $n=2$, 
we obtain P$_{\rm IV}$ \eqref{eqn:intro_P4} with the conditions \eqref{eqn:intro_fa_P4}.

From the affine Weyl group symmetry of system \eqref{eqn:differential_omega} given in Lemma \ref{lemma:action_omega_a}, 
we obtain that of P$_{\rm IV}$ as follows.
The actions of $\widetilde{W}(A_2^{(1)})=\langle s_0,s_1,s_2,\pi\rangle$ 
on the parameters $a_i$, $i=0,1,2$, are given by
\begin{equation}
 s_i(a_j)
 =\begin{cases}
 -a_j&\text{if\hspace{1em}}j=i,\\
 a_j+ a_i&\text{if\hspace{1em}}j=i\pm 1,\\
 a_j&\text{otherwise},
 \end{cases}\qquad
 \pi(a_i)=a_{i+1},
\end{equation}
where $i,j\in \bbZ/3\bbZ$,
while those on the variables $f_i$\,, $i=0,1,2$, are given by
\begin{equation}
 s_i(f_j)
 =\begin{cases}
 f_j+\cfrac{3a_i}{f_i}&\text{if\hspace{1em}}j=i-1,\\
 f_j-\cfrac{3a_i}{f_i}&\text{if\hspace{1em}}j=i+1,\\
 f_j&\text{if\hspace{1em}}j=i,
 \end{cases}\qquad
 \pi (f_i)=f_{i+1},
\end{equation}
where $i,j\in \bbZ/3\bbZ$.
Under these actions, the fundamental relations for $\widetilde{W}(A_2^{(1)})$ hold:
\begin{equation}\label{eqn:fundamental_A2}
 {s_i}^2=1,\quad 
 (s_is_{i\pm 1})^3=1,\quad
 (s_is_j)^2=1,\quad j\neq i,i\pm 1,\quad
 \pi s_i=s_{i+1}\pi,\quad
 \pi^3=1,
\end{equation}
where $i,j\in \bbZ/3\bbZ$.
The corresponding Dynkin diagram is given by Figure \ref{fig:dynkinA2}.
Before the discussion of the Lax representation of P$_{\rm IV}$,
let us consider the role of $\omega$-variables in the theory of the Painlev\'e IV equation.
%%%%%%%%%%%%%%%%%%%%%%%%%%%%%%
%% Lemma
%%%%%%%%%%%%%%%%%%%%%%%%%%%%%%
\begin{lemma}\label{lemma:omega_hamiltonian_P4}
The following relation holds:
\begin{equation}
 \omega_0=\sqrt{-3}\, h_{\rm IV}-\frac{(3+2t^2)t}{6},
\end{equation}
where $h_{\rm IV}$ is the Hamiltonian given by \eqref{eqn:intro_hamiltonian_P4}.
\end{lemma}
%%%%%%%%%%%%%%%%%%%%%%%%%%%%%%
%% proof
%%%%%%%%%%%%%%%%%%%%%%%%%%%%%%
\begin{proof}
Let
\begin{equation}
 c_0=\omega_0-\sqrt{-3}\ h_{\rm IV}+\frac{(3+2t^2)t}{6}.
\end{equation}
We can easily verify the following relations:
\begin{equation}\label{eqn;P4_c0_s_pi}
 s_i(c_0)=c_0,\quad i=0,1,2,\quad
 \pi(c_0)=c_0.
\end{equation}
Moreover, using the relation
\begin{align}
 \omega_0'
 =&\frac{\omega_0'+\omega_1'-(\omega_1'+\omega_2')+\omega_2'+\omega_0'}{2}\notag\\
 =&-\frac{(f_0-t)(2t-f_0)-(f_1-t)(2t-f_1)-(f_2-t)(2t-f_2)+2(a_1-a_2)+1}{2},
\end{align}
we obtain 
\begin{equation}\label{eqn;P4_c0_t}
 \frac{{\rm d}c_0}{{\rm d}t}=0.
\end{equation}
Equations \eqref{eqn;P4_c0_s_pi} and \eqref{eqn;P4_c0_t} mean that $c_0$ is an arbitrary constant.
It is obvious that without loss of generality we can put $c_0=0$. 
Therefore, we have completed the proof.
\end{proof}
%%%%%%%%%%%%%%%%%%%%%%%%%%%%%%
%%%%%%%%%%%%%%%%%%%%%%%%%%%%%%

Therefore, Theorems \ref{theorem:p4_hamiltonian} and \ref{theorem:p4_Lax} follow from Lemma \ref{lemma:lax_Ui} with 
\begin{equation}
 n=2,\quad
 i=0,\quad
 \Phi_0=\Phi,
\end{equation}
and Lemma \ref{lemma:omega_hamiltonian_P4}.

%%%%%%%%%%%%%%%%%%%%%%%%%%%%%%
%% Figure 3
%%%%%%%%%%%%%%%%%%%%%%%%%%%%%%
\begin{figure}[t]
\begin{center}
\includegraphics[width=0.4\textwidth]{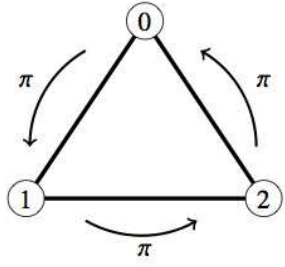}
\end{center}
\caption{Dynkin diagram of type $A_2^{(1)}$.}
\label{fig:dynkinA2}
\end{figure}
%%%%%%%%%%%%%%%%%%%%%%%%%%%%%%
%%%%%%%%%%%%%%%%%%%%%%%%%%%%%%

%%%%%%%%%%%%%%%%%%%%%%%%%%%%%%%%%%%%%%%%%%%%%%%%%%%%%%%%%%%
%% 4.2 Case $n=3$: the Painlev\'e V equation
%%%%%%%%%%%%%%%%%%%%%%%%%%%%%%%%%%%%%%%%%%%%%%%%%%%%%%%%%%%
\subsection{Case $n=3$: the Painlev\'e V equation}\label{subsection:P5}
In a similar manner to the case $n=2$ (see \S \ref{subsection:P4}), 
the case $n=3$ gives P$_{\rm V}$ \eqref{eqn:intro_P5} and its difference-differential Lax representation.

Let 
\begin{equation}
 f_0=\omega_1-\omega_3+t,\quad
 f_1=\omega_2-\omega_0+t,\quad
 f_2=\omega_3-\omega_1+t,\quad
 f_3=\omega_0-\omega_2+t.
\end{equation}
Then, we obtain P$_{\rm V}$ \eqref{eqn:intro_P5} and the conditions \eqref{eqn:intro_fa_P5} from system \eqref{eqn:differential_omega} and  the condition \eqref{eqn:sum_ai}.
The action of extend affine Weyl group $\widetilde{W}(A_3^{(1)})=\langle s_0,s_1,s_2,s_3,\pi\rangle$
on the parameters $a_i$, $i=0,\dots,3$, are given by
\begin{equation}
 s_i(a_j)
 =\begin{cases}
 -a_j&\text{if\hspace{1em}}j=i,\\
 a_j+ a_i&\text{if\hspace{1em}}j=i\pm 1,\\
 a_j&\text{otherwise},
 \end{cases}\qquad
 \pi(a_i)=a_{i+1},
\end{equation}
where $i,j\in\bbZ/4\bbZ$, while those on the variables  $f_i$, $i=0,\dots,3$, are given by
\begin{equation}
 s_i(f_j)
 =\begin{cases}
 f_j+\cfrac{4a_i}{f_i}&\text{if\hspace{1em}}j=i-1,\\
 f_j-\cfrac{4a_i}{f_i}&\text{if\hspace{1em}}j=i+1,\\
 f_j&\text{otherwise},
 \end{cases}\qquad
 \pi (f_i)=f_{i+1},
\end{equation}
where $i,j\in \bbZ/4\bbZ$.
Under these actions, $\widetilde{W}(A_3^{(1)})$ satisfies the following fundamental relations:
\begin{equation}\label{eqn:fundamental_A3}
 {s_i}^2=1,\quad 
 (s_is_{i\pm 1})^3=1,\quad
 (s_is_j)^2=1,\quad j\neq i,i\pm 1,\quad
 \pi s_i=s_{i+1}\pi,\quad
 \pi^4=1,
\end{equation}
where $i,j\in \bbZ/4\bbZ$.
The Dynkin diagram for $\widetilde{W}(A_3^{(1)})$ is given by Figure \ref{fig:dynkinA3}.
In a similar manner to the proof of Lemma \ref{lemma:omega_hamiltonian_P4}, we can prove the following lemma.
%%%%%%%%%%%%%%%%%%%%%%%%%%%%%%
%% Lemma
%%%%%%%%%%%%%%%%%%%%%%%%%%%%%%
\begin{lemma}\label{lemma:omega_hamiltonian_P5}
The following relation holds:
\begin{equation}
 \omega_0=\frac{16 h_{\rm V}-t^4-6 t^2-1}{8 t},
\end{equation}
where $h_{\rm V}$ is the Hamiltonian given by \eqref{eqn:intro_hamiltonian_P5}.
\end{lemma}
%%%%%%%%%%%%%%%%%%%%%%%%%%%%%%
%%%%%%%%%%%%%%%%%%%%%%%%%%%%%%

%%%%%%%%%%%%%%%%%%%%%%%%%%%%%%
%% Figure 4
%%%%%%%%%%%%%%%%%%%%%%%%%%%%%%
\begin{figure}[t]
\begin{center}
\includegraphics[width=0.4\textwidth]{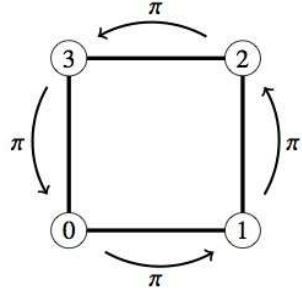}
\end{center}
\caption{Dynkin diagram of type $A_3^{(1)}$.}
\label{fig:dynkinA3}
\end{figure}
%%%%%%%%%%%%%%%%%%%%%%%%%%%%%%
%%%%%%%%%%%%%%%%%%%%%%%%%%%%%%

Therefore, Theorems \ref{theorem:p5_hamiltonian} and \ref{theorem:p5_Lax} follow from Lemma \ref{lemma:lax_Ui} with 
\begin{equation}
 n=3,\quad
 i=0,\quad
 \Phi_0=\Phi,
\end{equation}
and Lemma \ref{lemma:omega_hamiltonian_P5}.
%%%%%%%%%%%%%%%%%%%%%%%%%%%%%%%%%%%%%%%%%%%%%%%%%%%%%%%%%%%
%% 5. Concluding remarks
%%%%%%%%%%%%%%%%%%%%%%%%%%%%%%%%%%%%%%%%%%%%%%%%%%%%%%%%%%%
\section{Concluding remarks}\label{ConcludingRemarks}
In this paper, we have constructed the relation between the ABS equations and NY-system through the periodic type reduction.
Using this connection, we obtained the difference-differential Lax representations of P$_{\rm IV}$ and P$_{\rm V}$.
Moreover, we showed that the dependent variable of the system of ABS equations \eqref{eqn:standard_H1} can be reduced to the Hamiltonians of P$_{\rm IV}$ and P$_{\rm V}$.

An interesting future project is to investigate the relations between ABS equations and the other Painlev\'e equations (i.e., P$_{\rm VI}$, P$_{\rm III}$, P$_{\rm II}$, P$_{\rm I}$).
The results in this direction will be reported in forthcoming publications.
%%%%%%%%%%%%%%%%%%%%%%%%%%%%%%
%% Acknowledgment
%%%%%%%%%%%%%%%%%%%%%%%%%%%%%%
\subsection*{Acknowledgment}
The author would like to express his sincere thanks to Profs M. Noumi and Y. Yamada for inspiring and fruitful discussions.
I also appreciate the valuable comments from the referee which have improved the quality of this paper.
This research was supported by a grant \# DP160101728 from the Australian Research Council and JSPS KAKENHI Grant Number JP17J00092.
%%%%%%%%%%%%%%%%%%%%%%%%%%%%%%%%%%%%%%%%%%%%%%%%%%%%%%%
%%% Appendix
%%%%%%%%%%%%%%%%%%%%%%%%%%%%%%%%%%%%%%%%%%%%%%%%%%%%%%%
\appendix
%%%%%%%%%%%%%%%%%%%%%%%%%%%%%%%%%%%%%%%%%%%%%%%%%%%%%%%
%%% A. Proof of Lemma \ref{lemma:Lax_n+2system}
%%%%%%%%%%%%%%%%%%%%%%%%%%%%%%%%%%%%%%%%%%%%%%%%%%%%%%%
\section{Proof of Lemma \ref{lemma:Lax_n+2system}}
\label{section:appendix_proof_lemma_Lax}
In this section, we construct the Lax representation of system \eqref{eqn:standard_H1}
following the method given in \cite{BS2002:MR1890049,NijhoffFW2002:MR1912127,WalkerAJ:thesis,JNS:paper3}.

The key to constructing the Lax representation of system \eqref{eqn:standard_H1} is
to introduce a virtual direction from the lattice $\bbZ^{n+2}$, where system \eqref{eqn:standard_H1} is assigned, to the multi-dimensionally consistent integer lattice $\bbZ^{n+3}$.
Then, system \eqref{eqn:standard_H1} can be extended to the following system of \PDE s:
\begin{equation}\label{eqn:standard_H1_n+3}
 (W-W_{\bi\,\bj})(W_\bi-W_\bj)+\al^{(i)}(l_i)-\al^{(j)}(l_j)=0,\quad 0\leq i<j\leq n+2,
\end{equation}
where $W=W(l_0,\dots,l_{n+1},l_{n+2})$.
Here, $W(l_0,\dots,l_{n+1},0)=u(l_0,\dots,l_{n+1})$ is the dependent variable of system \eqref{eqn:standard_H1}.
Distinguish $u(l_0,\dots,l_{n+1})$ from $v(l_0,\dots,l_{n+1}):=W(l_0,\dots,l_{n+1},1)$.
Then, each of equations between $u=u(l_0,\dots,l_{n+1})$ and $v=v(l_0,\dots,l_{n+1})$:
\begin{equation}\label{eqn:rel_u_v}
 (u-v_\bi)(u_\bi-v)+\al^{(i)}(l_i)-\mu=0,\quad i=0,\dots,n+1,
\end{equation}
where $\mu=\al^{(n+2)}(0)$,
can be regarded as the first order discrete system of Riccati type of the quantity $v$, which is linearizable.
Indeed, substituting 
\begin{equation}
 v(l_0,\dots,l_{n+1})=\cfrac{F(l_0,\dots,l_{n+1})}{G(l_0,\dots,l_{n+1})}\,,
\end{equation}
in \eqref{eqn:rel_u_v} and dividing them into the numerators and the denominators,
we obtain the following linear systems:
\begin{equation}
  \psi_\bi
 =\begin{pmatrix}1&u\\0&1\end{pmatrix}
 \begin{pmatrix}0&\mu-\al^{(i)}(l_i)\\1&-u_\bi\end{pmatrix}
\psi,\quad
 i=0,\dots,n+1,
\end{equation}
where the vector $\psi=\psi(l_0,\dots,l_{n+1})$ is defined by
\begin{equation}
 \psi(l_0,\dots,l_{n+1})=\begin{pmatrix}F(l_0,\dots,l_{n+1})\\G(l_0,\dots,l_{n+1})\end{pmatrix}.
\end{equation}
We can easily verify that the compatibility conditions
\begin{equation}
 (\psi_\bi)_\bj=(\psi_\bj)_\bi,\quad 0\leq i<j\leq n+1,
\end{equation}
are equivalent to system \eqref{eqn:standard_H1}.
Finally, using the replacements \eqref{eqn:l0_special_direction} and
\begin{equation}
 \psi(l_0,\dots,l_{n+1})=\psi_{l_1,\dots,l_{n+1}}(t+l_0\ep),
\end{equation}
we have completed the proof of Lemma \ref{lemma:Lax_n+2system}.

%%%%%%%%%%%%%%%%%%%%%%%%%%%%%%%%%%%%%%%%%%%%%%%%%%%%%%%
%%% B. Proof of Lemma \ref{lemma:reduction}
%%%%%%%%%%%%%%%%%%%%%%%%%%%%%%%%%%%%%%%%%%%%%%%%%%%%%%%
\section{Proof of Lemma \ref{lemma:reduction}}
\label{section:appendix_proof_lemma_reduction}
In this section, we give a proof of Lemma \ref{lemma:reduction}.
 
By using \eqref{eqns:def_U_phi}, Equations \eqref{eqn:lpKdVs_u_1} and \eqref{eqn:lpKdVs_u_2} can be respectively rewritten as
\begin{subequations}\label{eqns:lpKdVs_U}
\begin{align}
\begin{split}\label{eqn:lpKdVs_U_1}
 &\left(U-\ol{U}_\bj+(\tilde{\al}_{l_1,\dots,l_{n+1}}-(\tilde{\al}_{l_1,\dots,l_{n+1}})_\bj)t-(\tilde{\al}_{l_1,\dots,l_{n+1}})_\bj\,\ep+\ep^{-1}\right)\\
 &\quad\times\left(\ol{U}-U_\bj+(\tilde{\al}_{l_1,\dots,l_{n+1}}-(\tilde{\al}_{l_1,\dots,l_{n+1}})_\bj)t+\tilde{\al}_{l_1,\dots,l_{n+1}}\ep-\ep^{-1}\right)\\
 &\qquad=\cfrac{t^2+2}{4}-\ep^{-2}+\al^{(j)}(l_j),\quad  j=1,\dots, n+1,
\end{split}\\
\begin{split}\label{eqn:lpKdVs_U_2}
 &\left(U-U_{\bi,\bj}+(\tilde{\al}_{l_1,\dots,l_{n+1}}-(\tilde{\al}_{l_1,\dots,l_{n+1}})_{\bi,\bj})t\right)
 \left(U_\bi-U_\bj+((\tilde{\al}_{l_1,\dots,l_{n+1}})_\bi-(\tilde{\al}_{l_1,\dots,l_{n+1}})_\bj)t\right)\\
 &\qquad=-\al^{(i)}(l_i)+\al^{(j)}(l_j),\quad 
 1\leq i<j\leq n+1,
\end{split}
\end{align}
\end{subequations}
while Equations \eqref{eqn:lax_psi_t} and \eqref{eqn:lax_psi_li} can be respectively rewritten as
\begin{subequations}\label{eqns:lax_phi}
\begin{align}
 &\ep^{-1}(\ol{\phi}-\phi)
 =\begin{pmatrix}1&U+\tilde{\al}_{l_1,\dots,l_{n+1}}t\\0&1\end{pmatrix}
 \begin{pmatrix}0&\frac{U-\ol{U}}{\ep}-\tilde{\al}_{l_1,\dots,l_{n+1}}+\frac{t^2+2}{4}+\mu\\1&-\ol{U}-\tilde{\al}_{l_1,\dots,l_{n+1}}(t+\ep)\end{pmatrix}
 \phi,\\
 &\phi_\bi
 =\begin{pmatrix}1&U+\tilde{\al}_{l_1,\dots,l_{n+1}}t\\0&1\end{pmatrix}
 \begin{pmatrix}0&-\al^{(i)}(l_i)+\mu\\1&-U_\bi-(\tilde{\al}_{l_1,\dots,l_{n+1}})_\bi\,t\end{pmatrix}
 \phi,\quad
 i=1,\dots,n+1,
\end{align}
\end{subequations}
where $U=U_{l_1,\dots,l_{n+1}}(t)$ and $\phi=\phi_{l_1,\dots,l_{n+1}}(t)$.
The periodic condition \eqref{eqn:periodc_condition} imposes the condition that 
\eqref{eqns:lpKdVs_U} is equivalent to  \eqref{eqns:lpKdVs_U}$_{\ol{1},\dots,\ol{n+1}}$.
From the condition that \eqref{eqn:lpKdVs_U_2} is equivalent to \eqref{eqn:lpKdVs_U_2}$_{\ol{1},\dots,\ol{n+1}}$, we obtain
\begin{equation}
 \al^{(i)}(l_i+1)-\al^{(i)}(l_i)=\al^{(j)}(l_j+1)-\al^{(j)}(l_j),\quad 1\leq i<j\leq n+1,
\end{equation}
which hold under the conditions of parameters \eqref{eqn:para_al_be}.
From the condition that \eqref{eqn:lpKdVs_U_1} is equivalent to \eqref{eqn:lpKdVs_U_1}$_{\ol{1},\dots,\ol{n+1}}$, we obtain the following condition:
\begin{equation}
 |\ep|\ll 1,
\end{equation}
which causes the continuum limit of systems \eqref{eqns:lpKdVs_U} and \eqref{eqns:lax_phi} to systems \eqref{eqns:differential_difference_U} and \eqref{eqns:difference_differential_phi}, respectively.
Therefore, we have completed the proof of Lemma \ref{lemma:reduction}.
%%%%%%%%%%%%%%%%%%%%%%%%%%%%%%%%%%%%%%%%%%%%%%%%%%%%%%%%%%%
%%% References 
%%%%%%%%%%%%%%%%%%%%%%%%%%%%%%%%%%%%%%%%%%%%%%%%%%%%%%%%%%%
\def\cprime{$'$} \def\cprime{$'$}

\end{document}